\documentclass[a4paper,11pt]{article}
\usepackage[ams]{}
\usepackage{amsthm}
\usepackage{hyperref}
\usepackage{amssymb}
\usepackage{graphicx}
\usepackage[all,arc]{xy}
\usepackage[theorem]{}

\theoremstyle{plain}			

\newtheorem{thm}{Theorem}[section]

\newtheorem{prop}[thm]{Proposition}

\theoremstyle{definition}		

\newtheorem{defn}[thm]{Definition}

\newtheorem{rmk}[thm]{Remark}

\begin{document}
\title{Jordan $C^{\ast}$-Algebras and Supergravity} 
\author{Michael Rios\footnote{email: mrios4@calstatela.edu}\\\\\emph{California State University, Los Angeles}\\\emph{Physics Graduate Program}
\\\emph{5151 State University Drive}\\\emph{Los Angeles, CA 90032, USA}  } \date{\today}\maketitle
\begin{abstract}
It is known that black hole charge vectors of $\mathcal{N}=8$ and magic $\mathcal{N}=2$ supergravity in four and five dimensions can be represented as elements of Jordan algebras of degree three over the octonions and split-octonions and their Freudenthal triple systems.  We show both such Jordan algebras are contained in the exceptional Jordan $C^{\ast}$-algebra and construct its corresponding Freudenthal triple system and single variable extension.  The transformation groups for these structures give rise to the complex forms of the U-duality groups for $\mathcal{N}=8$ and magic $\mathcal{N}=2$ supergravities in three, four and five dimensions.  
\\\\
$Keywords:$ Jordan $C^{\ast}$-algebras, Supergravity, U-duality.
\end{abstract}

\newpage
\tableofcontents
\section{Introduction}
\indent It was shown that extremal black holes in $\mathcal{N}\geq 2$, $D=3,4,5,6$ supergravities on symmetric spaces can be described by Jordan algebras and their corresponding Freudenthal triple systems, with the Bekenstein-Hawking entropy of the black holes given by algebraic invariants over these structures \cite{4}-\cite{23},\cite{28,31,42,54,57,58}.  Such supergravities are called homogeneous supergravities \cite{10}, and include $\mathcal{N}=8$, $D=4$ supergravity (M-theory on $T^7)$, described by the Jordan algebra of degree three over the split-octonions.  Of the four $\mathcal{N}=2$ supergravity theories defined by simple Jordan algebras of degree three (magic supergravities), all but one, the exceptional magic $\mathcal{N}=2$ supergravity, can be obtained by a consistent truncation of the maximal $\mathcal{N}=8$ supergravity \cite{12}. The exceptional magic $\mathcal{N}=2$ supergavity corresponds to the exceptional Jordan algebra of degree three over the octonion composition algebra \cite{24}-\cite{30},\cite{catto,45,51,53}, and it has been recently noted \cite{57} there are physical and mathematical obstacles which make its study appear less well motivated than the $\mathcal{N}=8$ case.\\
\indent In this note, we show the maximal $\mathcal{N}=8$ and exceptional magic $\mathcal{N}=2$ supergravities can be unified using the exceptional Jordan $C^{\ast}$-algebra, a Jordan algebra of degree three over the bioctonions, which contains both Jordan algebras of degree three over the octonions and split-octonions.  We construct its corresponding Freudenthal triple system and single variable extension and give support for such constructions in light of non-real solutions arising in $\mathcal{N}=8$ and magic $\mathcal{N}=2$, $D=4$ supergravities reduced to $D=3$. 
\section{$\mathcal{N}=8$ and Magic $\mathcal{N}=2$ Supergravity}
\indent We review the $\mathcal{N}\geq 2$ homogeneous supergravities in $D=3,4,5,6$ over Jordan algebras of degree two and three and give their corresponding U-duality groups and entropy expressions for black hole and string solutions.
\subsection{Black Strings in $D=6$}
\indent The $\mathcal{N}=8$ and magic $\mathcal{N}=2$, $D=6$ supergravities arise as uplifts of $\mathcal{N}=8$ and magic $\mathcal{N}=2$, $D=5$ supergravities with $n_V=27$, $n_V=15$, $n_V=9$ and $n_V=6$ vector fields \cite{21}.  They enjoy $SO(5,5)$, $SO(9,1)$, $SO(5,1)$, $SO(3,1)$ and $SO(2,1)$ U-duality symmetry since the $D=6$ vector multiplets in the Coulomb phase (after Higgsing) transform as spinors of dimension 16, 8, 4 and 2, respectively \cite{21}.  One can associate a black string solution with charges $q_I$ ($I=1,...,n_V$ for $n_V=10,6,4,3$) an element
\begin{equation}
J=\sum_{I=1}^{n}q^Ie_I=\left(\begin{array}{cc} r_1 & A \\ \overline{A}& r_2  \end{array}\right)\quad r_i\in\mathbb{R}, A\in\mathbb{A}
\end{equation}
of a Jordan algebra $J^{\mathbb{A}}_2$ of degree two, where the $e_I$ form the $n_V$-dimensional basis of the Jordan algebra over a composition algebra $\mathbb{A}=\mathbb{O}_s, \mathbb{O}, \mathbb{H}, \mathbb{C}, \mathbb{R}$.  Elements of $J^{\mathbb{A}}_2$ transform as the (dim $\mathbb{A}+2$) representation of $SL(2,\mathbb{A})$, the $\mathbf{10}$, $\mathbf{10}$, $\mathbf{6}$, $\mathbf{4}$, $\mathbf{3}$ of $SO(5,5)$, $SO(9,1)$, $SO(5,1)$, $SO(3,1)$ and $SO(2,1)$, respectively \cite{4}.  In all cases, the black string entropy is given by \cite{4}:
\begin{equation}
S_{D=5,BH}(J)=\pi\sqrt{|I_2(J)|}
\end{equation}
where
\begin{equation}
I_2(J)=\textrm{det}(J)=r_1r_2-A\overline{A}.
\end{equation}
The U-duality orbits are distinguished by rank via
\begin{equation}\begin{array}{rcl}
\textrm{Rank}\thinspace J=2\quad \textrm{iff}\quad I_2(J)\neq 0\quad\quad\quad\qquad\hspace{1pt} S\neq 0,\thinspace \textrm{1/4-BPS}\hfill\\
\textrm{Rank}\thinspace J=1\quad \textrm{iff}\quad I_2(J)=0,\thinspace J\neq 0\qquad S=0,\thinspace\textrm{1/2-BPS}\hfill.
\end{array}
\end{equation}

\subsection{Black Holes in $D=5$}
\indent In $D=5$, the $\mathcal{N}=8$ and magic $\mathcal{N}=2$ supergravities are coupled to 27, 15, 9, 6 vector fields with U-duality symmetry groups $E_{6(6)}$, $E_{6(-26)}$, $SU^\ast(6)$, $SL(3,\mathbb{C})$, $SL(3,\mathbb{R})$, for Jordan algebras of degree three over composition algebras $\mathbb{A}=\mathbb{O}_s, \mathbb{O}, \mathbb{H}, \mathbb{C}, \mathbb{R}$,  respectively \cite{4,13,14}.  The orbits of BPS black hole solutions were classified by Ferrara et al. \cite{13,22,57} by studying the underlying Jordan algebras of degree three under the actions of their reduced structure groups, $\textrm{Str}_0(J_3^{\mathbb{A}})$, which correspond to the U-duality groups of the $\mathcal{N}=8$ and magic $\mathcal{N}=2$, $D=5$ supergravities.  This is accomplished by associating a given black hole solution with charges $q_I$ ($I=1,...,n_V)$ an element
\begin{equation}
J=\sum_{I=1}^{n}q^Ie_I=\left(\begin{array}{ccc} r_1 & A_1 & \overline{A}_2 \\ \overline{A}_1 & r_2 & A_3 \\ A_2 & \overline{A}_3 & r_3  \end{array}\right)\quad r_i\in\mathbb{R}, A_i\in\mathbb{A}
\end{equation}
of a Jordan algebra of degree three $J_3^{\mathbb{A}}$ over a composition algebra $\mathbb{A}=\mathbb{O}_s, \mathbb{O}, \mathbb{H}, \mathbb{C}, \mathbb{R}$, where $e_I$ form a basis for the $n_V$-dimensional Jordan algebra.  This establishes a correspondence between Jordan algebras of degree three and the charge spaces of the extremal black hole solutions \cite{13}.  In all cases, the entropy of a black hole solution can be written \cite{4,13,23} in the form
\begin{equation}
S_{D=5,BH}(J)=\pi\sqrt{|I_3(J)|}
\end{equation}
where $I_3$ is the cubic invariant given by 
\begin{equation}
I_3(J)=\textrm{det}(J)
\end{equation}
and 
\begin{equation}
\textrm{det}(J)=r_1r_2r_3-r_1||A_1||^2-r_2||A_2||^2-r_3||A_3||^3+2\textrm{Re}(A_1 A_2 A_3).
\end{equation}
The U-duality orbits are distinguished by rank via
\begin{equation}\begin{array}{rcl}
\textrm{Rank}\thinspace J=3\quad \textrm{iff}\quad I_3(J)\neq 0\hspace{3pt}\qquad\quad\qquad S\neq 0,\thinspace \textrm{1/8-BPS}\hfill\\
\textrm{Rank}\thinspace J=2\quad \textrm{iff}\quad I_3(J)=0, J^{\natural} \neq 0\qquad S=0,\thinspace \textrm{1/4-BPS}\hfill\\
\textrm{Rank}\thinspace J=1\quad \textrm{iff}\quad J^{\natural} = 0,\thinspace\thinspace J\neq 0\thinspace\thinspace\quad\qquad S=0,\thinspace\textrm{1/2-BPS}\hfill
\end{array}
\end{equation}
where the quadratic adjoint map is given by
\begin{equation}
J^{\natural}=J\times J = J^2-\textrm{tr}(J)J+\frac{1}{2}(\textrm{tr}(J)^2-\textrm{tr}(J^2))I.
\end{equation}
\subsection{Black Holes in $D=4$}
\indent For $\mathcal{N}=8$ and magic $\mathcal{N}=2$, $D=4$ supergravities there is a correspondence between the field strengths of the vector fields and their magnetic duals and elements of a Freudenthal triple system (FTS) $\mathfrak{M}(J^{\mathbb{A}}_3)$ over a Jordan algebra of degree three, $J^{\mathbb{A}}_3$  \cite{14}.  The correspondence is explicitly:
\begin{equation}
\left(\begin{array}{ccc} F_{\mu\nu}^0 & & F_{\mu\nu}^i \\ & & \\ \tilde{F}_i^{\mu\nu} & & \tilde{F}_0^{\mu\nu} \end{array}\right) \Longleftrightarrow
 \left(\begin{array}{ccc} \alpha & & X \\ & & \\ Y & & \beta \end{array}\right)=\mathcal{X}\in \mathfrak{M}(J^{\mathbb{A}}_3),
\end{equation}
where $\alpha,\beta\in\mathbb{R}$ and $X,Y\in J^{\mathbb{A}}_3$.  $F_{\mu\nu}^i$ ($i=1,...,n_V$) denote the field strengths of the vector fields from $D=5$ and $F_{\mu\nu}^0$ is the $D=4$ graviphoton field strength coming from the $D=5$ graviton \cite{14}.
Using the correspondence, one can associate the entries of an FTS element with electric and magnetic charges $\{q_0,q_i,p^0,p^i\}\in\mathbb{R}^{2n_{V} + 2}$ of an $\mathcal{N}=8$ or magic $\mathcal{N}=2$, $D=4$ extremal black hole via the relations \cite{12,14}:
\begin{equation}
\alpha=p^0 \quad \beta=q_0 \quad X=p^ie_i \quad Y=q_ie^i.
\end{equation}
Setting $p^I=(\alpha,X)$ and $q_I=(\beta,Y)$ ($I=0,...,n_V$), the Bekenstein-Hawking entropy of extremal black hole solutions is given by \cite{12}:
\begin{equation}
S_{D=4,BH}(\mathcal{X})=\pi \sqrt{|I_4(p^I,q_I)|},
\end{equation}
where $I_4$ is the quartic invariant of the FTS, preserved by the automorphism groups $\textrm{Aut}(\mathfrak{M}(J^{\mathbb{A}}_3))$, the U-duality groups for the corresponding $D=4$ supergravities.  Explicitly, the U-duality groups are $E_{7(7)}$, $E_{7(-25)}$, $SO^\ast(12)$, $SU(3,3)$ and $Sp(6,\mathbb{R})$ for composition algebras algebras $\mathbb{A}=\mathbb{O}_s,\mathbb{O},\mathbb{H},\mathbb{C},\mathbb{R}$, respectively.  The U-duality orbits are given by the FTS rank via\footnote[1]{Explicit forms for $T(\mathcal{X},\mathcal{X},\mathcal{X})$ and $\Lambda(\mathcal{X},\mathcal{Y})$ are given in section 3.4.}
\begin{equation}\begin{array}{rcl}
\textrm{Rank}\thinspace \mathcal{X}=4\quad \textrm{iff}\quad I_4(\mathcal{X})\neq 0, \quad (I_4(\mathcal{X})< 0)\qquad\qquad\quad\hspace{6pt} \textrm{non-BPS}\hfill\\
\textrm{Rank}\thinspace \mathcal{X}=4\quad \textrm{iff}\quad I_4(\mathcal{X})\neq 0, \quad (I_4(\mathcal{X})> 0)\qquad\qquad\quad\hspace{6pt} \textrm{1/8-BPS}\hfill\\
\textrm{Rank}\thinspace \mathcal{X}=3\quad \textrm{iff}\quad T(\mathcal{X},\mathcal{X},\mathcal{X})\neq 0,\thinspace I_4(\mathcal{X})= 0\quad\quad\quad\quad\thinspace\thinspace\thinspace\thinspace \textrm{1/8-BPS}\hfill\\
\textrm{Rank}\thinspace \mathcal{X}=2\quad \textrm{iff}\quad \exists\mathcal{Y}\thinspace \Lambda(\mathcal{X},\mathcal{Y})\neq 0,\thinspace T(\mathcal{X},\mathcal{X},\mathcal{X})=0\qquad\thinspace\textrm{1/4-BPS}\hfill\\
\textrm{Rank}\thinspace \mathcal{X}=1\quad \textrm{iff}\quad \forall \mathcal{Y}\thinspace \Lambda(\mathcal{X},\mathcal{Y})=0,\thinspace \mathcal{X}\neq 0\qquad\qquad\qquad\hspace{3pt} \textrm{1/2-BPS}\hfill
\end{array}
\end{equation}
where 
\begin{equation}
\Lambda(\mathcal{X},\mathcal{Y})=3T(\mathcal{X},\mathcal{X},\mathcal{Y})+\mathcal{X}\{\mathcal{X},\mathcal{Y}\} .
\end{equation}
\subsection{Black Holes in $D=3$}
By extending the space of electric and magnetic charges $(p^I,q_I)\in\mathfrak{M}(J^{\mathbb{A}}_3)$ by a real variable, we recover a space $\mathfrak{T}(J^{\mathbb{A}}_3)$ acted on by the three-dimensional U-duality group $G_3=\textrm{Inv}(\mathfrak{T}(J^{\mathbb{A}}_3))$ \cite{28}.  $G_3$ is $E_{8(8)}$ or $E_{8(-24)}$, for dimensionally reduced $\mathcal{N}=8$ or exceptional magic $\mathcal{N}=2$, $D=4$ supergravity \cite{10,12}.  The quartic symplectic distance $d(\mathbf{X},\mathbf{Y})$ between any two solutions $\mathbf{X}=(\mathcal{X},x)$ and $\mathbf{Y}=(\mathcal{Y},y)$ in $\mathfrak{T}(J^{\mathbb{A}}_3)$ is given by:
\begin{equation}
d(\mathbf{X},\mathbf{Y})=I_4(\mathcal{X}-\mathcal{Y})-(x-y+\{\mathcal{X},\mathcal{Y}\})^2.
\end{equation}
The norm of an arbitrary solution $\mathbf{X}=(\mathcal{X},x)\in\mathfrak{T}(J^{\mathbb{A}}_3)$ is computed as:
\begin{equation}
\mathcal{N}(\mathbf{X})=d(\mathbf{X},0)=I_4(\mathcal{X})-x^2.
\end{equation}
The U-duality group $G_3$ leaves light-like separations $d(\mathbf{X},\mathbf{Y})=0$ invariant.  Given an arbitrary  $\mathbf{X}=(\mathcal{X},x)\in\mathfrak{T}(J^{\mathbb{A}}_3)$, G\"{u}naydin et al. noticed \cite{28} non-real values arise for $x$ when zero-norm solutions contain a negative-valued quartic invariant $I_4$.  To remedy this, the representation space $\mathfrak{T}(J^{\mathbb{A}}_3)\sim \mathbb{R}^{57}$ is complexified, leading to a realization of $E_{8}(\mathbb{C})$ on $\mathbb{C}^{57}$. 

\section{Bioctonions and Jordan $C^{\star}$-Algebras}
\subsection{Composition Algebras}
Let $V$ be a finite dimensional vector space over a field $\mathbb{F}=\mathbb{R},\mathbb{C}$.  An $algebra$ $structure$ on $V$ is a bilinear map
\begin{eqnarray}
V\times V \rightarrow V \\
(x,y)\mapsto x\bullet y.
\nonumber
\end{eqnarray}
A $composition$ $algebra$ is an algebra $\mathbb{A}=(V,\bullet)$, admitting an identity element, with a non-degenerate quadratic form $\eta$ satisfying
\begin{eqnarray}
\forall x,y\in\mathbb{A}\quad \eta(x\bullet y)=\eta(x)\eta(y).
\end{eqnarray}
If $\exists x\in\mathbb{A}$ such that $x\neq 0$ and $\eta(x)=0$, $\eta$ is said to be $isotropic$ and gives rise to a $split$ $composition$ $algebra$.  When $\forall x\in\mathbb{A}$, $x\neq 0$, $\eta(x)\neq 0$, $\eta$ is $anisotropic$ and yields a $composition$ $division$ $algebra$ \cite{52}.
\begin{prop}
A finite dimensional vector space $V$ over $\mathbb{F}=\mathbb{R},\mathbb{C}$ can be endowed with a composition algebra structure if and only if $\textrm{dim}_{\mathbb{F}}(V)=1,2,4,8$.  If $\mathbb{F}=\mathbb{C}$, then for a given dimension all composition algebras are isomorphic.  For $\mathbb{F}=\mathbb{R}$ and $\textrm{dim}_{\mathbb{F}}(V)=8$ there are only two non-isomorphic composition algebras: the octonions $\mathbb{O}$ for which $\eta$ is anisotropic and the split-octonions $\mathbb{O}_s$ for which $\eta$ is isotropic and of signature $(4,4)$.  Moreover for all composition algebras, the quadratic form $\eta$ is uniquely defined by the algebra structure.
\end{prop}
\begin{proof}
See prop. 1.8.1., section 1.10 and corollary 1.2.4 in the book by Springer and Veldkamp \cite{52}.
\end{proof}

\subsection{Bioctonions}

The $bioctonion$ algebra $\mathbb{O}_{\mathbb{C}}$ is a composition algebra of dimension $8$ over $\mathbb{C}$, defined as the complexification of the octonion algebra $\mathbb{O}$ \cite{24}
\begin{equation}
\mathbb{O}_{\mathbb{C}}=\mathbb{O}\otimes\mathbb{C}=\{\psi=\varphi_1+i\varphi_2\thinspace\thinspace |\thinspace\thinspace\varphi_i\in\mathbb{O},\thinspace i^2=-1\}
\end{equation}
where the imaginary unit `$i$' is assumed to commute with all imaginary basis units $e_j$ ($j=1,2,...,7$) of $\mathbb{O}$.  The $octonionic$ $conjugate$ of an element of $\mathbb{O}_{\mathbb{C}}$ is taken to be
\begin{equation}
\overline{\psi}=\overline{\varphi}_1+i\overline{\varphi}_2
\end{equation}
with which we define a quadratic form $\eta:\mathbb{O}_{\mathbb{C}}\rightarrow\mathbb{C}$
\begin{equation}
\eta(\psi)=\psi\overline{\psi}.
\end{equation}
By application of the Moufang identities for the octonions \cite{25}, it can be shown that $\forall\psi\in\mathbb{O}_{\mathbb{C}}\quad\eta(\psi_1\psi_2)=\eta(\psi_1)\eta(\psi_2)$, making $\mathbb{O}_{\mathbb{C}}$ a composition algebra over $\mathbb{C}$.

\begin{table}
\begin{center}
  \begin{tabular}{ | c | c | c | c | c | c | c | c | }
    \hline
        & $e_1$ & $e_2$ & $e_4$ &  $ie_7$ & $ie_3$ & $ie_6$  & $ie_5$ \\ \hline
    $e_1$ & $-1$  & $e_4$ & $-e_2$ & $-ie_3$ & $ie_7$ & $-ie_5$ & $ie_6$ \\ \hline
    $e_2$ & $-e_4$  & $-1$ & $e_1$ & $-ie_6$ & $ie_5$ & $ie_7$ & $-ie_3$ \\ \hline
    $e_4$ & $e_2$  & $-e_1$ & $-1$ & $-ie_5$ & $-ie_6$ & $ie_3$ & $ie_7$ \\ \hline
    $ie_7$ & $ie_3$  & $ie_6$ & $ie_5$ & $1$ & $e_1$ & $e_2$ & $e_4$ \\ \hline
    $ie_3$ & $-ie_7$  & $-ie_5$ & $ie_6$ & $-e_1$ & $1$ & $e_4$ & $-e_2$ \\ \hline
    $ie_6$ & $ie_5$  & $-ie_7$ & $-ie_3$ & $-e_2$ & $-e_4$ & $1$ & $e_1$ \\ \hline
    $ie_5$ & $-ie_6$ & $ie_3$ & $-ie_7$ & $-e_4$ & $e_2$ & $-e_1$ & $1$ \\ \hline
  \end{tabular}
\end{center}
\caption{A split-octonion subalgebra of $\mathbb{O}_{\mathbb{C}}$}
\label{tab:sploctotable}
\end{table}
\begin{prop}
The octonion $\mathbb{O}$ and split-octonion $\mathbb{O}_s$ composition algebras are real subalgebras of the bioctonion algebra $\mathbb{O}_{\mathbb{C}}$.
\end{prop}
\begin{proof}
As $\mathbb{O}_{\mathbb{C}}=\mathbb{O}\otimes\mathbb{C}$, the octonion algebra is taken to be
\begin{displaymath}
\mathbb{O}=\{\psi_r\in\mathbb{O}_{\mathbb{C}}\thinspace\thinspace |\thinspace\thinspace \psi_r=\varphi+i0\quad\varphi\in\mathbb{O}\}
\end{displaymath}
where the quadratic form over $\mathbb{O}_{\mathbb{C}}$ reduces to $\eta(\psi_r)=\psi_r\overline{\psi}_r=\varphi\overline{\varphi}\in\mathbb{R}$, which is the usual anisotropic quadratic form for which $\mathbb{O}$ is a composition algebra.

\indent For the split-octonion case, we denote a basis for $\mathbb{O}_{\mathbb{C}}$ via the set $\{e_i,ie_i\}$ where $e_i$ ($i=0,1,...,7$) form a basis for $\mathbb{O}$ as a real vector space satisfying
\begin{displaymath}
e_0=1
\end{displaymath}
\begin{displaymath}
e_i^2=-1\quad(i=1,2,...,7)
\end{displaymath}
\begin{displaymath}
e_{i+1}e_{i+2}=e_{i+4}=-e_{i+2}e_{i+1}\quad\textrm{(mod 7)}
\end{displaymath}
\begin{displaymath}
e_{i+2}e_{i+4}=e_{i+1}=-e_{i+4}e_{i+2}\quad\textrm{(mod 7)}.
\end{displaymath}
We now choose eight basis units from $\{e_i,ie_i\}$ consisting of a quaternion basis in $\mathbb{O}$, for example, $1,e_1,e_2,e_4$ and taking $ie_i$ such that ($i\neq 0,1,2,4$).  Consider the real vector subspace $W\subset\mathbb{O}_{\mathbb{C}}$ spanned by these basis units
\begin{displaymath}
W=\{\psi_s\in\mathbb{O}_{\mathbb{C}}\thinspace\thinspace |\thinspace\thinspace \psi_s=a_0+a_1e_1+a_2e_2+a_4e_4+i(a_3e_3+a_5e_5+a_6e_6+a_7e_7)\}
\end{displaymath}
Using the multiplicative properties of the octonions, one can construct a multiplication table for the basis units of $W$ where $W$ is seen to be closed (see Table \ref{tab:sploctotable}).  Conjugation in $W$ is induced by octonionic conjugation in $\mathbb{O}_\mathbb{C}$
\begin{equation}
\overline{\psi_s}=a_0-a_1e_1-a_2e_2-a_4e_4+i(-a_3e_3-a_5e_5-a_6e_6-a_7e_7).
\end{equation}
The quadratic form for $\mathbb{O}_{\mathbb{C}}$ then reduces to
\begin{equation}
\eta(\psi_s)=\psi_s\overline{\psi}_s=a_0^2+a_1^2+a_2^2+a_4^2-a_3^2-a_5^5-a_6^2-a_7^2\in\mathbb{R}
\end{equation}
which is isotropic and of signature $(4,4)$.  Hence, $W$ forms a $\textrm{dim}_{\mathbb{R}}(W)=8$ split composition algebra and by Prop. 3.1 must be isomorphic to the algebra of split-octonions $\mathbb{O}_s$.
\end{proof}
It was shown by Shukuzawa \cite{shuk1} that $G_2(\mathbb{C})$ acts transitively on the space of all elements having the same norm in $\mathbb{O}_{\mathbb{C}}$.  We shall recall some useful theorems from Shukuzawa \cite{shuk1} here, which classify orbits for elements of $\mathbb{O}_{\mathbb{C}}$ and its real subalgebras $\mathbb{O}$ and $\mathbb{O}_{s}$.  For notational convenience we set $ie_i=e_i'$, for the case of the split-octonions.
\begin{thm}
Any non-zero element $\psi\in\mathbb{O}_{\mathbb{C}}$ can be transformed to the following canonical form by some element of $G_2(\mathbb{C})$:\\\\
If $\eta(\psi)\neq 0$:
\begin{equation}
\psi=(a_0+ia_1)e_i\quad (i=1,2,...,7)\thinspace\thinspace (a_0 > 0\thinspace\thinspace\textrm{or}\thinspace\thinspace a_0=0,\thinspace a_1 > 0),
\end{equation}
If $\eta(\psi)= 0$: 
\begin{equation}
\psi=e_i+ie_j\quad (i\neq j,\thinspace\thinspace i,j=1,2,...,7).
\end{equation}
Moreover, their orbits in $\mathbb{O}_{\mathbb{C}}$ under $G_2(\mathbb{C})$ are distinct, and the union of all their orbits and $\{ 0 \}$ is the whole space $\mathbb{O}_{\mathbb{C}}$.
\end{thm}
\begin{rmk}
The canonical zero-norm orbit elements of $\mathbb{O}_{\mathbb{C}}$ generate a non-associative Grassmann algebra \cite{catto}, satisfying $\psi_i^2=0$, $\psi_i\psi_j=-\psi_j\psi_i$ and $\psi_i\overline{\psi}_j=-\psi_j\overline{\psi}_i$. 
\end{rmk}

\begin{thm}
Any element $\varphi\in\mathbb{O}$ can be transformed to the following canonical form by some element of $G_2$:\\
\begin{equation}
\varphi=a_0e_i\quad (i=1,2,...,7)\thinspace\thinspace (a_0=\sqrt{\eta(\varphi)} \geq 0).
\end{equation}
Moreover, their orbits in $\mathbb{O}$ under $G_2$ are distinct, and the union of all their orbits and $\{0\}$ yields the whole space $\mathbb{O}$.
\end{thm}

\begin{thm}
Any non-zero element $\psi_s\in\mathbb{O}_{s}$ can be transformed to the following canonical form by some element of $G_{2(2)}$:\\\\
If $\eta(\psi_s) > 0$:
\begin{equation}
\psi_s=a_0e_i\quad (i=1,2,...,7)\thinspace\thinspace (a_0=\sqrt{\eta(\varphi)} > 0)
\end{equation}
If $\eta(\psi_s) < 0$:
\begin{equation}
\psi_s=a_0e_i'\quad (i=1,2,...,7)\thinspace\thinspace (a_0=\sqrt{-\eta(\varphi)} > 0)
\end{equation}
If $\eta(\psi_s)= 0$: 
\begin{equation}
\psi=e_i+e_j'\quad (i\neq j,\thinspace\thinspace i,j=1,2,...,7).
\end{equation}
Moreover, their orbits in $\mathbb{O}_{s}$ under $G_{2(2)}$ are distinct, and the union of all their orbits and $\{ 0 \}$ is the whole space $\mathbb{O}_{s}$.
\end{thm}

\subsection{The Exceptional Jordan $C^{\star}$-Algebra}
\begin{defn}
(Kaplansky) Let $\mathcal{A}$ be a complex Banach space and a complex Jordan algebra equipped with an involution $\ast$.  Then $\mathcal{A}$ is a $Jordan$ $C^{\ast}$-$algebra$ if the following conditions are satisfied
\begin{equation}\begin{array}{rcl}
||x\circ y||\leq ||x|| \thinspace||y||\quad\forall x,y\in\mathcal{A}\hfill\\
||z||=||z^{\ast}||\quad\forall z\in\mathcal{A}\hfill\\
||\{zz^{\ast}z\}||=||z||^3\quad\forall z\in\mathcal{A}\hfill,
\end{array}
\end{equation}
where the Jordan triple product is given by
\begin{equation}\begin{array}{rcl}
\{xyz\}=(x\circ y)\circ z + (y\circ z)\circ x - (z\circ x)\circ y.\hfill
\end{array}
\end{equation}
\end{defn}
\begin{thm}
Each JB-algebra is the self-adjoint part of a unique Jordan $C^{\ast}$-algebra.
\end{thm}
\begin{proof}
The proof is given by Wright \cite{wright}, using the existence of an exceptional Jordan $C^{\star}$-algebra whose self-adjoint part is the exceptional Jordan algebra $J^{\mathbb{O}}_3$.
\end{proof}

The $exceptional$ $Jordan$ $C^{\ast}$-$algebra$ is the complexification of the exceptional Jordan algebra, given by
\begin{equation}
J^{\mathbb{O}_{\mathbb{C}}}_3=\{X=A+iB\thinspace\thinspace|\thinspace\thinspace A,B\in J^{\mathbb{O}}_3,\thinspace i^2=-1\},
\end{equation}
where the imaginary unit `$i$' is assumed to commute with all imaginary basis units $e_j$ ($j=1,...,7$) of $\mathbb{O}$.  A general element of the algebra takes the form 
\begin{equation}
X = \left(\begin{array}{ccc}z_1 & \psi_1 & \overline{\psi}_2 \\ \overline{\psi}_1 & z_2 & \psi_3 \\ \psi_2 & \overline{\psi}_3 & z_3 \end{array}\right) \qquad \qquad z_i \in \mathbb{C} \quad \psi_i \in \mathbb{O}_{\mathbb{C}}
\end{equation}
where it is seen $J^{\mathbb{O}_{\mathbb{C}}}_3$ is a Jordan algebra of degree three over the bioctonions.  As the bioctonions contain both the octonion and split-octonion algebras, $J^{\mathbb{O}_{\mathbb{C}}}_3$ contains $J^{\mathbb{O}}_3$ and $J^{\mathbb{O}_{s}}_3$.  One can define two types of involution for $J^{\mathbb{O}_{\mathbb{C}}}_3$
\begin{equation}
X^{\ast} =(\overline{X}^c)^T=(A-iB)^T,
\end{equation}
\begin{equation}
X^{\dagger} =(\overline{X})^T=(\overline{A}+i\overline{B})^T
\end{equation}
differing by the use of either complex or octonionic conjugation of the entries.  Using the complex involution, along with the spectral norm, $J^{\mathbb{O}_{\mathbb{C}}}_3$ becomes a Jordan $C^{\ast}$-algebra. Under this involution, only elements of the exceptional Jordan algebra $J^{\mathbb{O}}_3$ are self-adjoint.  Under the involution using octonionic conjugation, all elements of $J^{\mathbb{O}_{\mathbb{C}}}_3$ are self-adjoint.  Moreover, under this involution, the trace bilinear form $J^{\mathbb{O}_{\mathbb{C}}}_3 \times J^{\mathbb{O}_{\mathbb{C}}}_3 \rightarrow \mathbb{C}$
\begin{equation}
\langle X,Y\rangle=\textrm{tr}(X\circ Y^{\dagger})=\textrm{tr}(X\circ Y)
\end{equation}
is complex valued, as is required in later constructions.  The \textit{Freudenthal Product} $J^{\mathbb{O}_{\mathbb{C}}}_3 \times J^{\mathbb{O}_{\mathbb{C}}}_3 \rightarrow J^{\mathbb{O}_{\mathbb{C}}}_3$ is defined using the trace bilinear form as
\begin{equation}
X\times Y = X \circ Y - \frac{1}{2}(Y\textrm{tr}(X)+X\textrm{tr}(Y))+\frac{1}{2}(\textrm{tr}(X)\textrm{tr}(Y)-\textrm{tr}(X\circ Y))I
\end{equation}
An important special case yields the \textit{quadratic adjoint map}
\begin{equation}
X^{\natural}=X\times X = X^2-\textrm{tr}(X)X+\frac{1}{2}(\textrm{tr}(X)^2-\textrm{tr}(X^2))I,
\end{equation}
We can use the Freudenthal and Jordan product to define the \textit{cubic form}
\begin{equation}\label{15}
(X,Y,Z)=\textrm{tr}(X\circ (Y\times Z)). 
\end{equation}
A special case of this cubic form is
\begin{equation}
(X,X,X)=\textrm{tr}(X\circ(X \times X))
\end{equation}
Using the cubic form, one can express the determinant as 
\begin{equation}
\textrm{det}(X)=\frac{1}{3}\textrm{tr}(X\circ(X \times X))=N(X).
\end{equation}
where $N(X)$ denotes the \textit{cubic norm} of $X$.  
The structure group $\textrm{Str}(J^{\mathbb{O}_{\mathbb{C}}}_3)$, is comprised of all linear bijections on $J^{\mathbb{O}_{\mathbb{C}}}_3$ that leave the cubic norm (hence determinant) invariant up to a constant scalar multiple
\begin{equation}
N(s(X))=c\thinspace N(X)\qquad \forall s\in\textrm{Str}(J^{\mathbb{O}_{\mathbb{C}}}_3).
\end{equation}
The reduced structure group, $\textrm{Str}_0(J^{\mathbb{O}_{\mathbb{C}}}_3)=E_6(\mathbb{C})$ \cite{51}, consists of the transformations for which $c=1$ and contains the U-duality groups $E_{6(6)}$, $E_{6(-26)}$, of the $\mathcal{N}=8$ and exceptional magic $\mathcal{N}=2$, $D=5$ supergravities, respectively. 
\subsection{Freudenthal Triple System}
We follow the Freudenthal construction of Krutelevich et al. \cite{4,5,27}.  Given the exceptional Jordan $C^{\ast}$-algebra $J^{\mathbb{O}_{\mathbb{C}}}_3$, one can construct its corresponding Freudenthal triple system (FTS) by defining the vector space $\mathfrak{M}(J^{\mathbb{O}_{\mathbb{C}}}_3)$:
\begin{equation}
\mathfrak{M}(J^{\mathbb{O}_{\mathbb{C}}}_3)=\mathbb{C}\oplus\mathbb{C}\oplus J^{\mathbb{O}_{\mathbb{C}}}_3\oplus J^{\mathbb{O}_{\mathbb{C}}}_3.
\end{equation}
A general element $\mathcal{X}\in\mathfrak{M}(J^{\mathbb{O}_{\mathbb{C}}}_3)$ can be expressed as
\begin{equation}
\mathcal{X}=\left(\begin{array}{cc} \alpha & X \\ Y & \beta  \end{array}\right)\quad \alpha,\beta\in\mathbb{C}\quad X,Y\in J^{\mathbb{O}_{\mathbb{C}}}_3
\end{equation}
The FTS comes equipped with a non-degenerate bilinear antisymmetric quadratic form $\{\mathcal{X},\mathcal{Z}\}:\mathfrak{M}(J^{\mathbb{O}_{\mathbb{C}}}_3)\times\mathfrak{M}(J^{\mathbb{O}_{\mathbb{C}}}_3)\rightarrow\mathbb{C}$,
\begin{equation}\begin{array}{rcl}
\{\mathcal{X},\mathcal{Z}\}=\alpha\delta-\beta\gamma+\textrm{tr}(X\circ W)-\textrm{tr}(Y\circ Z) \\
\textrm{where}\quad \mathcal{X}=\left(\begin{array}{cc} \alpha & X \\ Y & \beta  \end{array}\right),\quad \mathcal{Z}=\left(\begin{array}{cc} \gamma & Z \\ W & \delta  \end{array}\right),
\end{array}
\end{equation}
a quartic form $q:\mathfrak{M}(J^{\mathbb{O}_{\mathbb{C}}}_3)\rightarrow \mathbb{C}$,
\begin{equation}
q(\mathcal{X})=-2[\alpha\beta-\textrm{tr}(X\circ Y)]^2-8[\alpha N(X)+\beta N(Y)-\textrm{tr}(X^{\natural}\circ Y^{\natural})]
\end{equation}
and a trilinear triple product $T:\mathfrak{M}(J^{\mathbb{O}_{\mathbb{C}}}_3)\times\mathfrak{M}(J^{\mathbb{O}_{\mathbb{C}}}_3)\times\mathfrak{M}(J^{\mathbb{O}_{\mathbb{C}}}_3)\rightarrow\mathfrak{M}(J^{\mathbb{O}_{\mathbb{C}}}_3)$,
\begin{equation}
\{T(\mathcal{X},\mathcal{Y},\mathcal{W}),\mathcal{Z}\}=q(\mathcal{X},\mathcal{Y},\mathcal{W},\mathcal{Z})
\end{equation}
where $q(\mathcal{X},\mathcal{Y},\mathcal{W},\mathcal{Z})$ is the linearization of $q(\mathcal{X})$ such that $q(\mathcal{X},\mathcal{X},\mathcal{X},\mathcal{X})=q(\mathcal{X})$.  Useful identities include \cite{27}
\begin{equation}\begin{array}{lll}
T(\mathcal{X}, \mathcal{X}, \mathcal{X})=\\(-\alpha^2\beta+\alpha\textrm{tr}(X\circ Y)-2N(Y),\quad\alpha\beta^2-\beta\textrm{tr}(X\circ Y)+2N(X), \\
\qquad\qquad\qquad\qquad\hspace{20pt} 2 Y\times X^{\natural}-2\beta Y^{\natural}-(\textrm{tr}(X\circ Y)-\alpha\beta)X,\\
\qquad\qquad\qquad\qquad\hspace{7pt} -2 X\times Y^{\natural}+2\alpha X^{\natural}+(\textrm{tr}(X\circ Y)-\alpha\beta)Y).
\end{array}
\end{equation}
\begin{equation}\begin{array}{lll}
\Lambda(\mathcal{X},\mathcal{Y})=3T(\mathcal{X}, \mathcal{X}, \mathcal{Y})+\{\mathcal{X},\mathcal{Y}\}\mathcal{X}=\\
(-(3\alpha\beta-\textrm{tr}(X\circ Y))\gamma+2\textrm{tr}((\alpha X - Y^{\natural})\circ W),\\
(3\alpha\beta-\textrm{tr}(X\circ Y))\delta-2\textrm{tr}((\beta Y - X^{\natural})\circ Z), \\
(3\alpha\beta-\textrm{tr}(X\circ Y))Z-2(\beta Y - X^{\natural})\times W+2(\alpha X - Y^{\natural})\delta-2Q(\mathcal{X})Z,\\
-(3\alpha\beta-\textrm{tr}(X\circ Y))W-2(\alpha X - Y^{\natural})\times Z+2(\beta Y - X^{\natural})\gamma-2Q(\mathcal{X}')W).
\end{array}
\end{equation}
where the Jordan triple product is used to define the quadratic polynomial
\begin{equation}
Q(\mathcal{X})=(\alpha\beta-\textrm{tr}(X\circ Y))Z+2\{X,Z,Y\}
\end{equation}
and 
\begin{equation}
Q(\mathcal{X}')W=Q(\mathcal{X})W-2Q(\mathcal{X})I\circ W.
\end{equation}

The automorphism group $\textrm{Aut}(\mathfrak{M}(J^{\mathbb{O}_{\mathbb{C}}}_3))=E_7(\mathbb{C})$ is the set of all transformations leaving the quadratic form and quartic form $q(\mathcal{X})=I_4(\mathcal{X})$ invariant.  Following Krutelevich \cite{27}, four types of transformations in $\textrm{Aut}(\mathfrak{M}(J^{\mathbb{O}_{\mathbb{C}}}_3))$ are:\\
For any $C\in J^{\mathbb{O}_{\mathbb{C}}}_3$, $\Phi(C)$:
\begin{equation}
\left(\begin{array}{cc} \alpha & X \\ Y & \beta  \end{array}\right)\mapsto\left(\begin{array}{cc} \alpha+\textrm{tr}(Y\circ C)+\textrm{tr}(X\circ C^{\natural})+\beta N(C) & X+\beta C \\ Y+X\times C+\beta C^{\natural} & \beta  \end{array}\right)
\end{equation}
For any $D\in J^{\mathbb{O}_{\mathbb{C}}}_3$, $\Psi(D)$:\\
\begin{equation}
\left(\begin{array}{cc} \alpha & X \\ Y & \beta  \end{array}\right)\mapsto\left(\begin{array}{cc} \alpha & X+Y\times D+\alpha D^{\natural} \\ Y+\alpha D & \beta+\textrm{tr}(X\circ D)+\textrm{tr}(Y\circ D^{\natural})+\alpha N(D)  \end{array}\right)
\end{equation}
For any $s\in\textrm{Str}(J^{\mathbb{O}_{\mathbb{C}}}_3)$ and $c\in\mathbb{C}$ s.t. $N(s(Z))=c\thinspace N(Z)$:
\begin{equation}
\Omega:\quad\left(\begin{array}{cc} \alpha & X \\ Y & \beta  \end{array}\right)\mapsto\left(\begin{array}{cc} c^{-1}\alpha & s(X) \\ s^{{*}^{-1}}(Y) & c\beta  \end{array}\right)
\end{equation}
where $s^{*}$ is the adjoint to $s$ with respect to the trace bilinear form.\\  Lastly, we have:
\begin{equation}
\Upsilon:\quad\left(\begin{array}{cc} \alpha & X \\ Y & \beta  \end{array}\right)\mapsto\left(\begin{array}{cc} -\beta & -Y \\ X & \alpha  \end{array}\right).
\end{equation}
When matrices $C$ and $D$ above are rank one (i.e., $C^{\natural}=0$, $D^{\natural}=0$), the transformations $\phi$ and $\psi$ simplify to:
\begin{equation}
\Phi(C):\left(\begin{array}{cc} \alpha & X \\ Y & \beta  \end{array}\right)\mapsto\left(\begin{array}{cc} \alpha+\textrm{tr}(Y\circ C) & X+\beta C \\ Y+X\times C & \beta  \end{array}\right)
\end{equation}
\begin{equation}
\Psi(D):\left(\begin{array}{cc} \alpha & X \\ Y & \beta  \end{array}\right)\mapsto\left(\begin{array}{cc} \alpha & X+Y\times D \\ Y+\alpha D & \beta+\textrm{tr}(X\circ D)  \end{array}\right).
\end{equation}

\begin{rmk}
Consider the space $\mathfrak{M}(\textrm{diag}(J^{\mathbb{O}_{\mathbb{C}}}_3))\subset\mathfrak{M}(J^{\mathbb{O}_{\mathbb{C}}}_3)$ with elements
\begin{equation}
\left(\begin{array}{cc} \alpha & X \\ Y & \beta  \end{array}\right)=\left(\begin{array}{cc} \alpha & (z_1,z_2,z_3) \\ (z_4,z_5,z_6) & \beta  \end{array}\right)\quad \alpha,\beta\in\mathbb{C}\quad X,Y\in J^{\mathbb{O}_{\mathbb{C}}}_3.
\end{equation}
This space is equivalent to the Freudenthal triple system $\mathfrak{M}(\mathfrak{J}_{\mathbb{C}})$ employed by Borsten et al. \cite{5} as the representation space of three qubits.  It would be interesting to find a quantum information interpretation for the full Freudenthal triple system $\mathfrak{M}(J^{\mathbb{O}_{\mathbb{C}}}_3)$.
\end{rmk}
\subsection{Extended Freudenthal Triple System}
Following G\"{u}naydin et al. \cite{28, 31}, we construct a vector space $\mathfrak{T}(J^{\mathbb{O}_{\mathbb{C}}}_3)$ by extending the FTS $\mathfrak{M}(J^{\mathbb{O}_{\mathbb{C}}}_3)$ by an extra complex coordinate:
\begin{equation}
\mathfrak{T}(J^{\mathbb{O}_{\mathbb{C}}}_3)=\mathfrak{M}(J^{\mathbb{O}_{\mathbb{C}}}_3)\oplus\mathbb{C}.
\end{equation}   
For brevity, we refer to the vector space $\mathfrak{T}(J^{\mathbb{O}_{\mathbb{C}}}_3)$ as the extended Freudenthal triple system (EFTS) over $J^{\mathbb{O}_{\mathbb{C}}}_3$.  Vectors in $\mathfrak{T}(J^{\mathbb{O}_{\mathbb{C}}}_3)$ are written in the form $\mathbf{X}=(\mathcal{X},\tau)$, where $\mathcal{X}\in\mathfrak{M}(J^{\mathbb{O}_{\mathbb{C}}}_3)$ belongs to the underlying FTS and $\tau\in\mathbb{C}$ is the extra complex coordinate.  The quartic symplectic distance $d(\mathbf{X},\mathbf{Y})$ between any two points $\mathbf{X}=(\mathcal{X},\tau)$ and $\mathbf{Y}=(\mathcal{Y},\kappa)$ in $\mathfrak{T}(J^{\mathbb{O}_{\mathbb{C}}}_3)$ is given by
\begin{equation}
d(\mathbf{X},\mathbf{Y})=q(\mathcal{X}-\mathcal{Y})-(\tau-\kappa+\{\mathcal{X},\mathcal{Y}\})^2.
\end{equation}
The norm of an arbitrary element $\mathbf{X}=(\mathcal{X},\tau)\in\mathfrak{T}(J^{\mathbb{O}_{\mathbb{C}}}_3)$ takes the form
\begin{equation}
\mathcal{N}(\mathbf{X})=d(\mathbf{X},0)=q(\mathcal{X})-\tau^2.
\end{equation}
The group leaving light-like separations $d(\mathbf{X},\mathbf{Y})=0$ invariant, the quasiconformal group of the EFTS \cite{10,13,28,31}, is now $\textrm{Inv}(\mathfrak{T}(J^{\mathbb{O}_{\mathbb{C}}}_3))=E_8(\mathbb{C})$.  The action of the Lie algebra of $\textrm{Inv}(\mathfrak{T}(J^{\mathbb{O}_{\mathbb{C}}}_3))$ on an arbitrary element of the EFTS $\mathbf{X}=(\mathcal{X},\tau)$ is given by \cite{28,31}:
\begin{equation}\begin{array}{ccl}
K(\mathcal{X})=0\quad U(\mathcal{X})=\mathcal{W}\quad S(\mathcal{X})=T(\mathcal{W},\mathcal{Z},\mathcal{X})\quad \mathcal{W},\mathcal{Z}\in\mathfrak{M}(J^{\mathbb{O}_{\mathbb{C}}}_3)\\
K(\tau)=2z\quad U(\tau)=\{\mathcal{W},\mathcal{X}\}\quad S(\tau)=2\{\mathcal{W},\mathcal{Z}\}\tau\quad z\in\mathbb{C}\\
\widetilde{U}(\mathcal{X})=\frac{1}{2}T(\mathcal{X},\mathcal{W},\mathcal{X})-\mathcal{W}\tau\\
\widetilde{U}(\tau)=-\frac{1}{6}\{T(\mathcal{X},\mathcal{X},\mathcal{X}),\mathcal{W}\}+\{\mathcal{X},\mathcal{W}\}\tau\\
\widetilde{K}(\mathcal{X})=-\frac{1}{6}zT(\mathcal{X},\mathcal{X},\mathcal{X})+z\mathcal{X}\tau\\
\widetilde{K}(\tau)=\frac{1}{6}z\{T(\mathcal{X},\mathcal{X},\mathcal{X}),\mathcal{X}\}+2z\tau^2.\\
\end{array}
\end{equation}
As noted by G\"{u}naydin, Koepsell and Nicolai \cite{28}, this gives a realization of $E_8(\mathbb{C})$ on $\mathbb{C}^{57}$, which remedies the problem encountered when $\mathcal{N}(\mathbf{X})=0$ and the quartic invariant $q$ takes negative values, forcing $\tau$ to have non-real solutions.

\section{Conclusion}
We have shown that $\mathcal{N}=8$ and $\mathcal{N}=2$ supergravity theories based on the octonions and split-octonions can be mathematically unified using the bioctonion composition algebra and its corresponding exceptional Jordan $C^{\ast}$-algebra, $J^{\mathbb{O}_{\mathbb{C}}}_3$.  Moreover, by constructing a Freudenthal triple system and its single variable extension over $J^{\mathbb{O}_{\mathbb{C}}}_3$, problematic solutions in $D=3$ were resolved.  The exceptional Jordan $C^{\ast}$-algebra $J^{\mathbb{O}_{\mathbb{C}}}_3$ and its Freudenthal triple system also proved useful in supporting the three qubit entanglement classification of Borsten et al.\cite{5}.\\
\indent Surely, there are further applications for Jordan algebraic structures based on the bioctonions, and it is interesting to consider the direct physical interpretations of such structures in M-theory \cite{35,43},\cite{47}-\cite{51}.  Along these lines, it is essential to consider structures over the integral bioctonions, enabling the study of the discrete U-duality orbits of $E_6(\mathbb{C})_{\mathbb{Z}}$, $E_7(\mathbb{C})_{\mathbb{Z}}$ and $E_8(\mathbb{C})_{\mathbb{Z}}$, with applications to topological strings \cite{10,12}, quantum information theory \cite{1}-\cite{9},\cite{32}-\cite{41} and automorphic black hole partition functions.

\end{document}